\documentclass[runningheads]{llncs}
\usepackage{amsfonts}
\usepackage{amsmath}
\usepackage{booktabs} % For pretty tables
\usepackage{caption} % For caption spacing
\usepackage{graphicx}
\usepackage{pgfplots}
\usepackage[all]{nowidow}
\usepackage[utf8]{inputenc}
\usepackage{tikz}
\usetikzlibrary{er,positioning,bayesnet}
\usepackage{multicol}
\usepackage{algpseudocode,algorithm,algorithmicx}
\usepackage[inline]{enumitem} % Horizontal lists
\newtheorem{assumption}{Assumption}

\definecolor{blue}{HTML}{1F77B4}
\definecolor{orange}{HTML}{FF7F0E}
\definecolor{green}{HTML}{2CA02C}

\pgfplotsset{compat=1.14}

\begin{document}
\title{On Optimal Control of Discounted Cost Infinite-Horizon Markov Decision Processes Under  Local State Information Structures}
\titlerunning{Local State Information Markov Decision Process}
% If the paper title is too long for the running head, you can set
% an abbreviated paper title here
%
\author{Guanze Peng\inst{1} \and
Veeraruna Kavitha\inst{2} \and
Quanyan Zhu\inst{1}}
%
%\authorrunning{F. Author et al.}
% First names are abbreviated in the running head.
% If there are more than two authors, 'et al.' is used.
%
\institute{Dept. of Electrical and Computer Engineering, New York University
\email{\{guanze.peng,quanyan.zhu\}@nyu.edu}\\ \and
Industrial Engineering and Operation Research Department, Indian Institutes of Technology\\
\email{vkavitha@iitb.ac.in}}
\maketitle              % typeset the header of the contribution
\begin{abstract}
This paper investigates a class of optimal control problems associated with Markov  processes with local state information. The decision-maker has only a local access to a subset of a state vector information as often encountered in decentralized control problems in multi-agent systems. Under this information structure, part of the state vector cannot be observed. We leverage ab initio principles and find a new form of Bellman equations to characterize the optimal policies of the control problem under local information structures. The dynamic programming solutions feature a mixture of dynamics associated unobservable state components and the local state-feedback policy based on the observable local information. We further characterize the optimal local-state feedback policy using linear programming methods. To reduce the computational complexity of the optimal policy, we propose an approximate algorithm based on virtual beliefs to find a sub-optimal policy. We show the performance bounds on the sub-optimal solution and corroborate the results with numerical case studies. 

\end{abstract}
\section{INTRODUCTION}
The subject of Markov decision process (MDP) has been broadly explored in the area of robotics, wireless communication, and economics. In MDPs, the decision-maker is assumed to have complete state information. Notwithstanding, in many real world application, the direct observation of the state is either impossible or difficult to acquire (See \cite{sharma1997framework}, \cite{penggame}, \cite{huang2019continuous}). Therefore, partially observable Markov decision process (POMDP) becomes a standard framework where the decision-maker does not have direct access to the state information but indirect observations that are correlated with the true state. A substantial literature has been established over the past few decades, including \cite{puterman2014markov,altman1999constrained,krishnamurthy2016partially,sondik1978optimal}. %In standard POMDPs, the observation is deteriorated by exogenous randomness and the source of this randomness is assumed to be independent and identically distributed (i.i.d.) random variables. 

In standard POMDPs, the state information as a whole is taken as incompletely observable and the observations are statistically dependent on the state. In this work, we consider a class of problems where the state takes the form of a vector and its information can be partitioned into two components. One component contains a subset of states that are completely observable while the other component consists of a subset of states that are completely unobservable. This class of problem often arises from distributed multi-agent control systems, where one agent can only observe his own state while the state information of the others are not observable. We refer this class of problems as MDP under Local State Information or LSI-MDP, in short.

One difference between this class of problems and the classical POMDPs is that decision-maker of LSI-MDP has no information of a subset of states. As a result, the optimal control policy of the decision-maker takes the form of local-state feedback, which depends solely on the observable components of the state vector. % In other words, estimation of the unobservables can othe uncertainties of the state information comes from the  %the source of the uncertainty. In our work, we assume the uncertainty is induced by the system intrinsically. Hence, the uncertainty can be considered as part of the system state, which is affected by the actions made by the agent and evolves over time. 
We use two examples to motivate the  LSI-MDP model as follows.

\begin{enumerate}

	\item \textit{Team Optimization Problem and Multiagent System Problems}\\
	In both team optimization problem and multiagent systems, multiple agents make decisions based on their observations to optimize their objective functions, and the decisions can impact the state of the system, which is the aggregation of states of all the agents (See \cite{singh1994multiagent}, \cite{gupta2015existence}). If their objective functions are fully aligned, the problem becomes a team problem. If their objective functions are partially aligned with each other, the problem becomes a nonzero-sum game problem. Our work studies this problem from the perspective of a single agent in which the agent knows his own state but has no access to the states of other agents.
	
	\item \textit{Optimal Planning in Robotics}\\
	The robots plan the route or actions based on the observations or information it acquires (See \cite{kaelbling1998planning,parkan1999decision}). Nevertheless, due to the physical limitation of the sensors, there is no guarantee that the robots are capable of obtaining the complete observation of the state (See \cite{sharma1997framework}). Hence, the state can be divided into two parts: one part is observable and the other part is unobservable. As the unobservable part of the state is also influenced by the actions, this scenario coincides with our model.

\end{enumerate}

Specifically, our contributions can be summarized as follows:
\begin{itemize}
	\item We formulate an LSI-MDP problem and characterize the local-state feedback policy using the principle of optimality. We identify the connections with MDPs and POMDPs. 
	\item We show that the local-state feedback policy is characterized by a mixture of open-loop deterministic nonlinear system dynamics and a feedback solution arising from dynamic programming. % finding optimal Markovian decision policy in such dynamical systems is challenging by showing that this kind of system is a hybrid system of deterministic nonlinear system and MDP.
	\item We develop a method termed as \textit{Virtual Belief Method} to provide a suboptimal stationary local feedback policy. We can show that the worst-case performance degradation is bounded.
\end{itemize}

This paper is organized as the following, In Section II, we present the problem formulation and identify the relations of our framework with MDPs and POMDPs. In Section III, we use the principle of optimality to establish the associated Bellman-like equation. In Section IV, we propose a method to find suboptimal solutions. In Section V, we study several special cases regarding the structures of the system dynamics, cost function, and transition probabilities. It is shown that under some of the special cases, the method proposed in Section IV can yield optimal solution. In Section VI, we conclude this work and give possible directions of the future work.

\section{PROBLEM FORMULATION}
%As is customary in most of the optimization problems in a dynamical system, we commence by formally defining action space and state space. 
In this section, we present the problem formulation of the infinite-horizon discounted cost optimal control problem under local-state information. Let $\mathcal{A}$ be the finite action space and $\mathcal{X}$ be the finite state space. The state of the dynamical system considered in this work, is assumed to be a joint process of two \textit{substates}, one of which is called \textit{observable substate}. The $\textit{observable substate}$, which is denoted by $x_{\rm o}$, can be observed to the agent directly and utilized for the decision making. The other substate is called \textit{unobservable substate}, which is denoted by $x_{\rm u}$, cannot be obtained as observations by the agent. Thus, the state space is the Cartesian product of two state subspaces as follows:
\begin{equation*}
\mathcal{X}=\mathcal{X}_{\rm u}\times\mathcal{X}_{\rm o},
\end{equation*}
where $\mathcal{X}_{\rm o}$ contains all the possible observable substates $x_{\rm o}$ and $\mathcal{X}_{\rm u}$ contains all the possible unobservable substates $x_{\rm u}$. 

The stage cost function is assumed to be a nonnegative and bounded stationary function
\begin{equation*}
c(x,a)\ :\ \mathcal{X}\times\mathcal{A}\ \rightarrow\ \mathbb{R}^+.
\end{equation*}
The transition probability is given by a stationary function
\begin{equation*}
p(x^\prime|x,a)\ :\ \mathcal{X}_{\rm o}\times\mathcal{X}_{\rm u}\times\mathcal{A}\ \rightarrow\ [0,1]^{|\mathcal{X}|}.
\end{equation*}
More specifically, it can be written as $p(x_{\rm o}^\prime,x_{\rm u}^\prime|x_{\rm o},x_{\rm u},a)$.

In this work, we study the following criteria:
\begin{equation}\label{opt1}
V^{\pi}_{\alpha_{\rm u}}(x_{{\rm o},0}) =\mathbb{E}^{\pi}\left[\sum_{t=0}^\infty \beta^tc(x_{{\rm o},t},x_{{\rm u},t},a_t)\bigg|x_{{\rm o},0},\right],
\end{equation}
where $x_{{\rm o},t}$, $x_{{\rm u},t}$, and $a_t$ are the observable substate, unobservable substate, and action at time $t$, respectively. We aim to determine the policy which minimizes \eqref{opt1}. Here, $\beta$ is the discount factor and $0\leq \beta<1$. The distribution of the initial state $x_{{\rm u},0}$ is given by $\alpha_{\rm u} $,  $\pi$ a policy is a collection of the decision rules, and each decision rule is a mapping from the space of the history of states and actions to the action space. The agent only has access to the observable substate $x_{\rm o}$ at each time instant. Therefore, his decision can only be dependent on the observation history formed by $x_{\rm o}$. Formally, denote the state-action history of the original system at time $t$ as
\begin{equation*}
    \begin{aligned}
        h_t=\left\{x_{{\rm o},t},x_{{\rm o},t-1},...,x_{{\rm o},0},\alpha_{\rm u},a_{t-1},a_{t-2},...,a_{0}\right\},
    \end{aligned}
\end{equation*}
and 
\begin{equation*}
    h_0=\left\{x_{{\rm o},0},\alpha_{\rm u}\right\}.
\end{equation*}
Let $\mathcal{H}_t$ be the space of $h_t$. By definition, $\pi=\left\{d_t\right\}_t$ and $d_t: \mathcal{H}_t\rightarrow \Delta\left(\mathcal{A}\right)$.

Here, $x_{\rm u}$ can be regarded as {unobservable} uncertainty in the dynamic system. Thus, to cope with this uncertainty, we have the expectation in \eqref{opt1} that averages out the randomness induced by $x_{\rm u}$.

It is clear that our framework differs from the classical MDPs and POMDPs and there exists close connections between the LSI-MDP framework and these two models. To see this, we let the dynamical system evolve according to the following rule
\begin{equation*}
(X_{\rm o}^\prime,X_{\rm u}^\prime)=f(x_{\rm o},x_{\rm u},a,\Gamma),
\end{equation*}
where $\Gamma$ is an exogenous random variable.

\begin{assumption}
	There exists a deterministic function $g(\cdot,\cdot)$ such that at each time instant, 
	\begin{equation}\label{compose}
	X_{{\rm u},t} = g(X_{{\rm o},t},\Tilde{X}_{{\rm u},t}),
	\end{equation}
	where $X_{{\rm o},t}$ and $\Tilde{X}_{{\rm u},t}$ are conditionally independent conditioned on a given state-action history.
\end{assumption}
This assumption means that we can decompose the unobservable into two parts: the first part is correlated with the observable substate and second part is independent of the observable substate.

\begin{assumption}
	For every $x_{\rm o}\in\mathcal{X}_{\rm o}$, $g(x_{\rm o},\Tilde{x}_{\rm u})$ is an injective function.
\end{assumption}

This assumption implies that, given a pair of $(x_{\rm o},x_{\rm u})$, we can identify the value of $\tilde{x}_{\rm u}$ uniquely.

Next, we use the following theorem to construct a controlled Markov process in which $\{x_{{\rm u},t}\}_t$ is conditionally independent of $\{x_{{{\rm o}},t}\}_t$.

\begin{theorem} Under \textit{Assumptions 1 and 2}, for any given fixed policy, there exists a random process $\{\Tilde{X}_{{\rm u},t}\}_t$ which satisfies the following:  
\begin{itemize}
    \item[a)] it evolves (conditionally) independently of $\{{X}_{{\rm o},t}\}_t$, i.e., such that
    \begin{equation*}\label{Eqn_cond_ind}
        p(x_{\rm o}^\prime,x_{\rm u}^\prime|x_{\rm o},x_{\rm u},a) = p(x_{\rm o}^\prime|x_{\rm o},a)p(\tilde{x}_{\rm u}^\prime|\tilde{x}_{\rm u},a),
    \end{equation*}
    where
      $x_{\rm u} = g(x_{\rm o},\Tilde{x}_{\rm u})$ and $x^\prime_{\rm u}  = g(x_{\rm o},\Tilde{x}^\prime_{\rm u})$;
	and  
	\item[b)] one  can represent the state of the system as $({X}_{{\rm o}, t},\Tilde{X}_{{\rm u}, t})$ with the same amount of information;
	at time $t+1$, there exists a deterministic function $\tilde{f}$ such that
	\begin{equation*}
	\begin{aligned}
	(X_{{\rm o},t+1},X_{{\rm u},t+1})&=
 \Tilde{f}(X_{{\rm o},t},\Tilde{X}_{{\rm u},t},a,\Gamma).
	\end{aligned} 
	\end{equation*}
\end{itemize}
\end{theorem}
\begin{proof}
 See Appendix.A.
\end{proof}
 \qed

In view of the above theorem, in some of the future sections,  we focus on the class of MDPs with `conditionally independent' transition probabilities  as in the right hand side of equation (\ref{Eqn_cond_ind}).

To complete the earlier argument, here we discuss how our model is related to POMDPs and MDPs. In POMDP, the observation and the state are assumed to be statistically correlated. Our system formulation includes the case where $x_{{\rm u},t}$ and $x_{{\rm o},t}$ are conditionally independent and $x_{{\rm u},t}$ provides no information of $x_{{\rm u},t}$ at all (once actions are observed).
Hence it is a generalization of POMDP. 
When the relation between $x_{\rm o}$ and $x_{\rm u}$ can be described by a deterministic function, $\Bar{g}$, such that $x_{u}=\Bar{g}(x_{\rm o})$, then our framework reduces to a classical MDP, as $x_{\rm o}$ can represent the system state.

%%%%%%%%%%%%%%%%%%%%%%%%%%%%%%%%%%%%%%%%%%%%%%%%%%%%%%%%%%%%%%%%%%%%%%%%%%%%%%%%
\section{DYNAMIC PROGRAMMING   with Beliefs}
As the substate $x_{\rm u}$ cannot be observed, the agent can form belief over the unobservable state. Denote the belief at time $t$ by $b({x}_{{\rm u},t})$, which evolves (depending upon observation $x_{o,t+1}$) according to 
\begin{equation}\label{update}
	    b({x}_{{\rm u},t+1})=\frac{\sum_{x_{{\rm u},t} }p(x_{{\rm o},t+1},x_{{\rm u},t+1}|x_{{\rm o},t},x_{{\rm u},t},a_t)b(x_{{\rm u},t})}{\sum_{\hat{x}_{{\rm u},t},\hat{x}_{{\rm u},t+1}}p(x_{{\rm o},t+1},\hat{x}_{{\rm u},t+1}|x_{{\rm o},t},\hat{x}_{{\rm u},t},a_t)b(\hat{x}_{{\rm u},t})},
\end{equation}
and 
\begin{equation*}
    b(x_{{\rm u},0}) = \alpha_u(x_{{\rm u},0}).
\end{equation*}
With a slight abuse of notation, let $b_t$ be the belief vector at time $t$. Thus, with the state denoted by $(x_{\rm o}, b)$, the system is Markovian. We define the transition function of the belief state as 
\begin{equation}\label{deterministic}
b_{t+1} = T(x_{{\rm o},t+1},x_{{\rm o},t},b_t,a_t).
\end{equation}
We would like to point out that the belief state acts as a deterministic nonlinear subsystem.

Define
\begin{equation*}
\bar{c}(x_{\rm o},b,a)=\sum_{{x}_{\rm u}}b({x}_{\rm u}){c}(x_{\rm o},x_{\rm u},a).
\end{equation*}
Since $x_{\rm u}$'s are not observable and we can only form belief over $x_{\rm u}$. After taking expectation using the belief of $x_{\rm u}$, we define the new objective function 
\begin{equation}\label{opt3}
\bar{V}_{\alpha_{\rm u}}^{\bar{\pi}}(x_{{\rm o},0},b_0) =\mathbb{E}^{\bar{\pi}}\left[\sum_{t=0}^\infty \beta^t\bar{c}(x_{{\rm o},t},b_t,a_t)\bigg|x_{{\rm o},0} \right].
\end{equation}
As we have mentioned above, the new system whose state is $(x_{\rm o},b)$ is Markovian. Denote the set of Markovian deterministic policies in this new system by $\Pi_{\rm MD}$. That is, the decision at time $t$ is only dependent on the current state $(x_{{\rm o},t},b_t)$.\\

For the system whose state is $(x_{{\rm o}},b)$, the state-action history at time $t$ is given by
  \begin{equation*}
    \begin{aligned}
            \bar{h}_t=\left\{x_{{\rm o},t},x_{{\rm o},t-1},...,x_{{\rm o},0},b_t,b_{t-1},...,b_0,a_{t-1},a_{t-2},...,a_{0}\right\},
    \end{aligned}
\end{equation*}
and 
\begin{equation*}
     \bar{h}_0=\left\{x_{{\rm o},0},b_0\right\}.
\end{equation*}
It is worth noting that $\bar{h}_t$ provides the same information as $h_t$, as $b_t$ evolves according to the rule \eqref{update}. 

\begin{lemma} If a given pair of policies $\pi=\left\{d_t\right\}_t$ and $\bar{\pi}=\left\{\bar{d}_t\right\}_t$ satisfies that $\bar{d}_t(\bar{h}_t)=d_t(h_t)$, then
	\begin{equation*}
	V^{\pi}_{\alpha_{\rm u}}(x_{{\rm o},0})=\bar{V}_{\alpha_{\rm u}}^{\bar{\pi}}(x_{{\rm o},0},b_0).
	\end{equation*}
\end{lemma}
\begin{proof}
    Since $\bar{d}_t(\bar{h}_t)=d_t(h_t)$ at every time instant, they generate the same action sequence $\{a_t\}_t$. To keep it simple, we provide the proof for deterministic policies, and the proof goes through in a similar way for randomized policies. 

Equation	\eqref{opt1} can be written as 
	\begin{equation*}
	\begin{aligned}
	&\mathbb{E}^{\pi}\left[\sum_{t=0}^\infty \beta^tc(x_{{\rm o},t},x_{{\rm u},t},a_t)\bigg| x_{{\rm o},0}\right]\\
	&=  \sum_{x_{{\rm u},0}}b_0(x_{{\rm u},0})c(x_{{\rm o},0},x_{{\rm u},0},d_0(h_0)) \\
	&\quad\quad\quad\quad\quad\quad\quad\quad\quad +\mathbb{E}^{\pi}\left[\sum_{t=1}^\infty \beta^tc(x_{{\rm o},t},x_{{\rm u},t},a_t)\bigg|x_{{\rm o},0}\right]\\
	&=\bar{c}(x_{{\rm o},0},b_0,d_0(h_0))+\mathbb{E}^{\pi}\left[\sum_{t=1}^\infty \beta^tc(x_{{\rm o},t},x_{{\rm u},t},a_t)\bigg|x_{{\rm o},0}\right]\\
	&=\bar{c}(x_{{\rm o},0},b_0,d_0(h_0))+\beta\mathbb{E}^{\pi}\left[  c(x_{{\rm o},1},x_{{\rm u},1},a_1)|x_{{\rm o},0}\right]\\
	&\quad\quad\quad\quad\quad\quad\quad\quad\quad +\mathbb{E}^{\pi}\left[\sum_{t=2}^\infty \beta^tc(x_{{\rm o},t},x_{{\rm u},t},a_t)\bigg|x_{{\rm o},0}\right],
	\end{aligned}
	\end{equation*}
	where $b_0(x_{{\rm u},0})$ is  equal to the distribution of the initial substate $x_{{\rm u},0}$ and can be interpreted as the belief of the unobservable substate at time $0$. At time $1$, the agent observes the history $h_1=\{x_{{\rm o},1},x_{{\rm o},0},b_0,a_0\}$ and chooses action $a_1$. Thus, the expectation in the second term can be expressed as 
	\begin{equation}\label{exp1}
    \begin{aligned}
    	\mathbb{E}^{\pi}[  c(x_{{\rm o},1},x_{{\rm u},1},a_1)|&x_{{\rm o},0}]\\
    	&=\mathbb{E}\left[\mathbb{E}[c(x_{{\rm o},1},x_{{\rm u},1},d_1(h_1))|h_1]|x_{{\rm o},0}\right],
    \end{aligned}
	\end{equation}
	where the inner conditional expectation of the RHS term in \eqref{exp1} is
	\begin{equation*}
	\begin{aligned}
	&\mathbb{E}^{\pi}[c(x_{{\rm o},1},x_{{\rm u},1},a_1)|x_{{\rm o},1},x_{{\rm o},0},b_0,a_1,a_0]\\
	&=\sum_{x_{{\rm u},1}}\frac{\sum_{x_{{\rm u},0}}p(x_{{\rm o},1},x_{{\rm u},1}|x_{{\rm o},0},x_{{\rm u},0},a_0)b_0(x_{{\rm u},0})}{\sum_{\hat{x}_{{\rm u},0},\hat{x}_{{\rm u},1}}p(x_{{\rm o},1},\hat{x}_{{\rm u},1}|x_{{\rm o},0},\hat{x}_{{\rm u},0},a_0)b_0(\hat{x}_{{\rm u},0})}\\
	&\quad\quad\quad\quad\quad\quad\quad\quad\quad\quad\quad\quad\quad\quad\quad\quad \cdot c(x_{{\rm o},1},x_{{\rm u},1},d_1(h_1))\\
	&=\sum_{x_{{\rm u},1}}b_1(x_{{\rm u},1}){c}(x_{{\rm o},1},x_{{\rm u},1},d_1(h_1))\\
	&=\bar{c}(x_{{\rm o},1},b_1,d_1(h_1)).
	\end{aligned}
	\end{equation*}
	
	An important property of this conditional expectation is that it captures the information structure of the agent. For example, at time $t=1$, after observing $x_{{\rm o},1}$, the agent aims to choose an action to minimize an objective function, $\mathbb{E}[c(x_{{\rm o},1},x_{{\rm u},1},a_1)|x_{{\rm o},1},x_{{\rm o},0},b_0,a_0,a_1]$. The expectation is taken over is induced by the unobservable substate $x_{{\rm u},1}$. This expectation requires us to have the knowledge of the distribution of $x_{{\rm u},1}$. 
	And we use the following as the distribution of $x_{{\rm u},1},$
	\begin{equation*}
	    \frac{\sum_{x_{{\rm u},0}}p(x_{{\rm o},1},x_{{\rm u},1}|x_{{\rm o},0},x_{{\rm u},0},a_0)b_0(x_{{\rm u},0})}{\sum_{\hat{x}_{{\rm u},0},\hat{x}_{{\rm u},1}}p(x_{{\rm o},1},\hat{x}_{{\rm u},1}|x_{{\rm o},0},\hat{x}_{{\rm u},0},a_0)b_0(\hat{x}_{{\rm u},0})},
	\end{equation*}
	
	as $x_{{\rm u},0}$ is also unknown and is averaged out using $b_0$. The outer expectation is taken with respect to $x_{{\rm o},1}$ conditioned on $x_{{\rm o},0}$ using the marginal distribution $\sum_{x_{{\rm u},1}}p(x_{{\rm o},1},x_{{\rm u},1}|x_{{\rm o},0},x_{{\rm u},0},a_0)$ by averaging $x_{{\rm u},0}$ out with $b_0$.
	
	Progressing this way for any $n < \infty$ we have the following:
    \begin{equation*}
    \begin{aligned}
        \mathbb{E}^{\pi}\bigg[\sum_{t=0}^n \beta^tc(x_{{\rm o},t},&x_{{\rm u},t},a_t)\bigg| x_{{\rm o},0}\bigg] \\ 
        &=\mathbb{E}^{\bar \pi}\left[\sum_{t=0}^n \beta^t {\bar c}(x_{{\rm o},t},b_{t},a_t)\bigg| x_{{\rm o},0}\right].
    \end{aligned}
    \end{equation*}
    By boundedness of $c(x_{\rm o},x_{\rm u},a)$ (and because  $\beta < 1$), the result follows by dominated convergence
	theorem by letting $n \to \infty$. 
\end{proof}
\qed

\begin{theorem} 
	\begin{equation*}
	\inf_{\pi\in\Pi}V^\pi_{\alpha_{\rm u}}(x_{\rm o})=\inf_{\pi\in\Pi}\bar{V}_{\alpha_{\rm u}}^{\bar{\pi}}(x_{{\rm o},0},b_0)=\inf_{\pi_{\rm MD}\in\Pi_{\rm MD}}\bar{V}_{\alpha_{\rm u}}^{\pi_{\rm MD}}(x_{{\rm o},0},b_0).
	\end{equation*}
\end{theorem}
\begin{proof}
    	The first equality directly follows from \textit{Lemma 1}: given the same policy, $V^{\pi}_{\alpha_{\rm u}}(x_{\rm o})$ and $\bar{V}^{\pi}_{\alpha_{\rm u}}(x_{{\rm o},0},b_0)$ yield the same value. Thus, they are minimized simultaneously. 
	
	To prove second equality, we analyze the new system whose state is $(x_{\rm o},b)$ beforehand. This new system can be viewed as a special case of MDPs. In MDPs, the current state and the next state are correlated through some transition kernel. Define,
	\begin{equation*}
	    p(x_{\rm o}^\prime|x_{\rm o},b,a)=\sum_{x_{\rm u},x_{\rm u}^\prime}p(x_{\rm o}^\prime,x_{\rm u}^\prime|x_{\rm o},x_{\rm u},a)b(x_{\rm u})
	\end{equation*}
	In our case, we have  
	\begin{equation*}
	\begin{aligned}
	p(x_{\rm o}^\prime,b^\prime|x_{\rm o},b,a)&=p(x_{\rm o}^\prime,T(x_{\rm o}^\prime,x_{\rm o},b,a)|x_{\rm o},b,a)\\
	&=p(x_{\rm o}^\prime|x_{\rm o},b,a)\delta(b^\prime,T(x_{\rm o}^\prime,x_{\rm o},b,a)),
	\end{aligned}
	\end{equation*}
	where $\delta(\cdot,\cdot)$ is a Dirac delta function, as the belief evolves like a deterministic system. The rest of the proof is by \cite{puterman2014markov}.
\end{proof}
\qed

Using standard arguments of dynamic programming, we can write down the following Bellman equation:
\begin{equation}\label{DP1}
\begin{aligned}
u(x_{\rm o},b)=\min_{a}\bigg\{\bar{c}(x_{\rm o}&,b,a)\\
&+\beta\sum_{x^\prime_{\rm o},b^\prime}u(x_{\rm o}^\prime,b^\prime)p(x_{\rm o}^\prime,b^\prime|x_{\rm o},b,a)\bigg\}.
\end{aligned}
\end{equation}

We note that the system can be regarded as a mixture of two subsystems: one is MDP and the other is a nonlinear deterministic system. And the states of these two subsystems are intertwined through the transition probability.

Let $\alpha_{\rm o}(x_{\rm o})$ be the distribution of the initial observable substate $x_{\rm o}$. Moreover, let $\mathcal{B}$ be the \textit{reachable set} of the belief, which contains all the possible belief.  If $\mathcal{B}$ is infinite, then the number of constraints of linear programming formed by \eqref{DP1} are also infinite, even for finite state and action spaces. Therefore, solving this optimization problem is challenging using the classical linear programming method.

\section{VIRTUAL BELIEF METHOD}
In this section, we propose a method called \textit{virtual belief method}, which aims to approximate the system with an MDP and provide a suboptimal solution. We show that this proposed method reduces the complexity and yet guarantees the performance by a bounded term. 

In the virtual belief method, the agent is assumed to believe that at each time instant, the system is at ${x}_{\rm u}$ with probability $b_0({x}_{\rm u})$, which is equal to the distribution of the initial substate $x_{{\rm u},0}$. And this belief stays unaltered throughout the whole process. 

To proceed, let us formally define the virtual system constructed by this method. The transition probability is in this system given by
\begin{equation*}
\tilde{p}(x_{\rm o}^\prime|x_{\rm o},a)=\sum_{x_{\rm u}}p(x_{\rm o}^\prime|x_{\rm o},x_{\rm u},a)b_0(x_{\rm u}).
\end{equation*}
The new cost function now becomes
\begin{equation*}
\tilde{c}(x_{\rm o},a)=\sum_{{x}_{\rm u}}b_0({x}_{\rm u})c({x}_{\rm o},{x}_{\rm u},a).
\end{equation*}
The objective function in the new system is given by
\begin{equation*}
\tilde{V}^{\pi}_{\alpha_{\rm u}}(x_{{\rm o},0})= \mathbb{E}^{\pi}\left[\sum_{t=0}^\infty\beta^t\tilde{c}(x_{{\rm o},t},a_{t})\bigg| x_{{\rm o},0}\right].
\end{equation*}

It is straightforward to check that the system considered in the virtual belief method is a classical MDP. Following the procedures in \cite{puterman2014markov}, the Bellman equation associated with the MDP is 
\begin{equation}\label{virtualbellman}
\tilde{u}(x_{\rm o})=\min_{a}\left\{\tilde{c}(x_{\rm o},a)+\beta \sum_{x_{\rm o}^\prime}\tilde{u}(x_{\rm o}^\prime)\tilde{p}(x_{\rm o}^\prime|x_{\rm o},a)\right\}.
\end{equation}
To solve this MDP, we revisit the method of linear programming. Likewise, we begin with the primal linear programming.\\

\textbf{Primal LP$^{\prime}$ (Virtual Belief Model)}
\begin{equation*}
\begin{aligned}
\min_{\tilde{u}(x_{\rm o})}&\ \sum_{x_{\rm o}}\tilde{u}\left(x_{\rm o}\right){\alpha}_{\rm o}(x_{\rm o})\\
{\rm s.t.} &\ \tilde{c}(x_{\rm o},a)+\beta \sum_{x_{\rm o}^\prime}\tilde{u}(x_{\rm o}^\prime)
\tilde{p}(x_{\rm o}^\prime|x_{\rm o},a)\leq \tilde{u}(x_{\rm o}),\quad \forall\ x_{\rm o},a.
\end{aligned}
\end{equation*}

The corresponding dual LP is given by the following.\\

\textbf{Dual LP$^{\prime}$}
\begin{equation*}
\begin{aligned}
\max_{\tilde{y}(x_{\rm o},a)}\quad &\sum_{x_{\rm o},a}\tilde{y}(x_{\rm o},a)\tilde{c}(x_{\rm o},a)\\
{\rm s.t.}\quad 
& {\alpha}_{\rm o}(x^\prime_{\rm o})+\sum_{x_{\rm o},a}\beta \tilde{p}(x^\prime_{\rm o}|x_{\rm o},a)y(x_{\rm o},a)=\sum_{a}\tilde{y}(x^\prime_{\rm o},a),\\
&\qquad\qquad\qquad\qquad\qquad\qquad\qquad\qquad\qquad\ \forall\ x^\prime_{\rm o}\\
&\tilde{y}(x_{\rm o},a)\geq 0,\quad\forall\ x_{\rm o},a.
\end{aligned}
\end{equation*}

Both linear programmings above are solvable as they have finite constraints (with finite state and action spaces). To see how the disregard of the evolution of the belief process $\{b_t\}_t$ can deteriorate the optimization performance, we first define the operator acting on $\tilde{u}(x_{\rm o})$ as follows:
\begin{equation}\label{operator}
\begin{aligned}
\tilde{\mathcal{L}}\tilde{u}(x_{\rm o})=\min_{a}\bigg\{\tilde{c}(x_{\rm o},a)+\beta \sum_{x_{\rm o}^\prime,x_{\rm u}}\tilde{u}(x_{\rm o}^\prime)\tilde{p}(x_{\rm o}^\prime|x_{\rm o},a)\bigg\}.
\end{aligned}
\end{equation}
To make comparisons, we consider the bellman equation in a full-information setting. In this setting, both $x_{\rm o}$ and $x_{\rm u}$ are available for decision making, resulting in a classical MDP. The objective funtion in this setting is given by
\begin{equation*}
    V^\pi_{{\rm f},\alpha_{\rm u}}(x_{{\rm o},0})=\mathbb{E}^\pi\left[\sum_{t=0}^\infty\beta^t c(x_{{\rm o},t},x_{{\rm u},t},a_t) \bigg| x_{{\rm o},0}\right].
\end{equation*}
Let the value function in the full-information setting be $v(x_{\rm o},x_{\rm u})$, which satisfies the following fixed-point equation
\begin{equation}\label{fullinfo}
\begin{aligned}
v(x_{\rm o},x_{\rm u})=\min_{a}\bigg\{&{c}(x_{\rm o},x_{\rm u},a)\\
&+\beta\sum_{x^\prime_{\rm o},x^\prime_{\rm u}}v(x_{\rm o}^\prime,x_{\rm u}^\prime)p(x_{\rm o}^\prime,x_{\rm u}^\prime|x_{\rm o},x_{\rm u},a)\bigg\}.
\end{aligned}
\end{equation}

Define the operator acting on $\left\{v(x_{\rm o},x_{\rm o})\right\}_{x_{\rm o},x_{\rm u}}$ as  
\begin{equation}\label{fullinfoopt}
\begin{aligned}
&\mathcal{L}v(x_{\rm o},x_{\rm u})=\min_{a}\bigg\{{c}(x_{\rm o},x_{\rm u},a)\\
&\qquad\qquad\qquad\qquad+\beta\sum_{x^\prime_{\rm o},x_{\rm u}^\prime}v(x_{\rm o}^\prime,x_{\rm u}^\prime)p(x_{\rm o}^\prime,x_{\rm u}^\prime|x_{\rm o},x_{\rm u},a)\bigg\}.
\end{aligned}
\end{equation}
Also, the fixed-point equation \eqref{fullinfo} can be transformed to the following linear programming problems.\\

\textbf{Primal LP$^{\prime\prime}$ (Full information case)}\\
\begin{equation*}
\begin{aligned}
\min_{v}&\ \sum_{x_{\rm o}}{v}\left(x_{\rm o},x_{\rm u}\right){\alpha}_{\rm o}(x_{\rm o}){\alpha}_{\rm u}(x_{\rm u})\\
{\rm s.t.} &\ {c}(x_{\rm o},x_{\rm u},a)+\beta \sum_{x_{\rm o}^\prime,x_{\rm u}^\prime}{v}(x_{\rm o}^\prime,x_{\rm u}^\prime)
{p}(x_{\rm o}^\prime,x_{\rm u}^\prime|x_{\rm o},x_{\rm u},a)\\
&\quad\quad\quad\quad \quad\quad\quad\quad\quad\quad\quad\quad\leq {v}(x_{\rm o},x_{\rm u}),\quad \forall\ x_{\rm o},x_{\rm u},a.
\end{aligned}
\end{equation*}

And the corresponding dual LP is given by the following.\\

\textbf{Dual LP$^{\prime\prime}$}
\begin{equation*}
\begin{aligned}
\max_{{y}_{\rm f}}\quad &\sum_{x_{\rm o},x_{\rm u},a}{y}_{\rm f}(x_{\rm o},x_{\rm u},,a){c}(x_{\rm o},x_{\rm u},a)\\
{\rm s.t.}\quad 
& \sum_{x_{\rm o},x_{\rm u},a}\beta {p}(x_{\rm o}^\prime,x_{\rm u}^\prime|x_{\rm o},x_{\rm u},a){y}_{\rm f}(x_{\rm o},x_{\rm u},a)\\
&\quad\ \ =\sum_{a}{y}_{\rm f}(x^\prime_{\rm o},x^\prime_{\rm u},,a)-{\alpha}_{\rm o}(x^\prime_{\rm o}){\alpha}_{\rm u}(x^\prime_{\rm u}),\ \forall\ x^\prime_{\rm o},x^\prime_{\rm u}\\
&{y}_{\rm f}(x_{\rm o},x_{\rm u},,a)\geq 0,\quad\forall\ x_{\rm o},x_{\rm u},a.
\end{aligned}
\end{equation*}

Before we give the main theorem of this section, we present the following two propositions. \cite{puterman2014markov}.
\begin{proposition}
	The operator defined in \eqref{operator} is a contraction mapping and it has the following properties:
	\begin{enumerate}
		\item if $\tilde{u}\geq \tilde{\mathcal{L}}\tilde{u}$, then $\tilde{u}\geq \tilde{u}^*$;
		\item if $\tilde{u}\leq \tilde{\mathcal{L}}\tilde{u}$, then $\tilde{u}\leq \tilde{u}^*$;
		\item if $\tilde{u}= \tilde{\mathcal{L}}\tilde{u}$, then  $\tilde{u}= \tilde{u}^*$.
	\end{enumerate}
\end{proposition}

\begin{proposition}
	The operator defined in \eqref{fullinfoopt} is a contraction mapping and it has the following properties:
	\begin{enumerate}
		\item if $v\geq {\mathcal{L}}v$, then $v\geq v^*$;
		\item if $v\leq {\mathcal{L}}v$, then $v\leq v^*$;
		\item if $v= {\mathcal{L}}v$, then  $v= v^*$.
	\end{enumerate}
\end{proposition}

\begin{theorem} (\textbf{Comparison between Full information and Virtual Belief models})\\
	If the transition probability can be decomposed as
	\begin{equation}
	\label{hypothesis}
	p(x_{\rm o}^\prime,x_{\rm u}^\prime|x_{\rm o},x_{\rm u},a)=p(x_{\rm o}^\prime|x_{\rm o},a)p(x_{\rm u}^\prime|x_{\rm u},a),
	\end{equation}
	then 
    \begin{equation*}
	\sup_{x_{\rm o, 0}}\left |	\inf_{\pi\in\Pi}\tilde{V}^\pi_{\alpha_{\rm u}}(x_{\rm o,0})  - \inf_{\pi\in\Pi}  V^\pi_{{\rm f},\alpha_{\rm u}}(x_{{\rm o},0}) \right| 
    \end{equation*}
	is bounded by 
    \begin{equation*}
        C=\max\left\{\frac{-\underline{C}}{1-\beta},\frac{\bar{C}}{1-\beta}\right\},
    \end{equation*}
    where
	\begin{equation*}
	\bar{C}:=\max_{x_{\rm o},{x}_{\rm u},{x}_{\rm u}^\prime,a}\left\{{c}(x_{\rm o},{x}_{\rm u},a)-{c}(x_{\rm },{x}_{\rm u}^\prime,a)\right\},
	\end{equation*}
	and
	\begin{equation*}
	\underline{C}:=\min_{x_{\rm o},{x}_{\rm u},{x}_{\rm u}^\prime,a}\left\{{c}(x_{\rm o},{x}_{\rm u},a)-{c}(x_{\rm },{x}_{\rm u}^\prime,a)\right\}.
	\end{equation*}
\end{theorem}
\begin{proof}
	Let $\tilde{u}^*(x_{\rm o})$ and $v^*(x_{\rm o},{x}_{\rm u})$ be the optimal solution to the virtual belief model and the full information setting, respectively. Then, the corresponding value functions of these models are given by $\sum_{x_{\rm o}}\tilde{u}^*(x_{\rm o}){\alpha}_{\rm o}(x_{\rm o})$ and $\sum_{x_{\rm o},{x}_{\rm u}}v^*(x_{\rm o},{x}_{\rm u}){\alpha}_{\rm o}(x_{\rm o}){\alpha}_{\rm u}({x}_{\rm u})$, respectively. Let $\tilde{a}^*$ be the optimal action (which depends upon $x_o$) that achieves the minimum in \eqref{virtualbellman} and ${a}^*$ be the optimal action (which depends upon $x_o, x_u$) that achieves the minimum in \eqref{fullinfoopt}. 
	
	With an abuse of notation, we define for any $x_{\rm u}$:
	\begin{equation*}
	\mathcal{L} \tilde{u}^*(x_{\rm o}):= 
	\mathcal{L}\tilde{\tilde{u}}^*(x_{\rm o}, x_{\rm u}),
	\end{equation*}
	where $\tilde{\tilde{u}}^*(x_{\rm o}, x_u):=\tilde{u}^*(x_{\rm o}).$
	%
	%We need to justify %$\mathcal{L}\tilde{u}^*(x_{\rm o})$.
	By definition and the given hypothesis, for any   $(x_{\rm o},x_{\rm u})$ (with $a^* = a^* (x_o, x_u)$),
    \begin{equation*}
    \begin{aligned}
    \mathcal{L}\tilde{u}^*(x_{\rm o})&={c}(x_{\rm o},x_{\rm u},a^*)
    +\beta\sum_{x^\prime_{\rm o},x_{\rm u}^\prime}\tilde{u}^*(x^\prime_{\rm o})p(x_{\rm o}^\prime,x_{\rm u}^\prime|x_{\rm o},x_{\rm u},a^*)\\
    &={c}(x_{\rm o},x_{\rm u},a^*)
    +\beta\sum_{x^\prime_{\rm o}}\tilde{u}^*(x^\prime_{\rm o})p(x_{\rm o}^\prime|x_{\rm o},a^*).
    \end{aligned}
    \end{equation*}
Further from (\ref{operator}) (for any $a^*(x_o,x_u)$)
\begin{eqnarray*}
\tilde{\mathcal{L}}\tilde{u}^*(x_{\rm o})
&\le & \tilde{c}(x_{\rm o},a^*)+\beta \sum_{x_{\rm o}^\prime}\tilde{u}^*(x_{\rm o}^\prime)\tilde{p}(x_{\rm o}^\prime|x_{\rm o},a^*)\\ \mbox{ and under hypothesis (\ref{hypothesis}) } \hspace{-10mm}&& \\
\tilde{p}(x_{\rm o}^\prime|x_{\rm o},a) &=&
{p}(x_{\rm o}^\prime|x_{\rm o},a).
\end{eqnarray*}  
Thus we have: 
	\begin{equation*}
	\begin{aligned}
	\tilde{u}^*(x_{\rm o})-\mathcal{L}\tilde{u}^*(x_{\rm o})&=\tilde{\mathcal{L}}\tilde{u}^*(x_{\rm o})-\mathcal{L}\tilde{u}^*(x_{\rm o})\\
	&\leq \max_{x_{\rm u}}\left\{\tilde{c}(x_{\rm o},{a}^*)-{c}(x_{\rm o},x_{\rm u},{a}^*)\right\}\\
	&\leq \max_{x_{\rm o},{x}_{\rm u},{x}_{\rm u}^\prime,a}\left\{{c}(x_{\rm o},{x}^\prime_{\rm u},a)-{c}(x_{\rm o},{x}_{\rm u},a)\right\}\\
	&=\bar{C}.
	\end{aligned}
	\end{equation*}

	On the other hand, now using $\tilde{a}^*$ we have
	\begin{equation*}
	\begin{aligned}
	\tilde{u}^*(x_{\rm o})-\mathcal{L}\tilde{u}^*(x_{\rm o})&=\tilde{\mathcal{L}}\tilde{u}^*(x_{\rm o})-\mathcal{L}\tilde{u}^*(x_{\rm o})\\
	&\geq \min_{x_{\rm u}}\left\{\tilde{c}(x_{\rm o},\tilde{a}^*)-{c}(x_{\rm o},x_{\rm u},\tilde{a}^*)\right\}\\
	&\geq \min_{x_{\rm o},{x}_{\rm u},{x}_{\rm u}^\prime,a}\left\{{c}(x_{\rm o},{x}^\prime_{\rm u},a)-{c}(x_{\rm o},{x}_{\rm u},a)\right\}\\
	&=\underline{C}.
	\end{aligned}
	\end{equation*}
	
	First consider the case where $\tilde{u}^*(x_{\rm o})\geq \mathcal{L}\tilde{u}^*(x_{\rm o})$.  Then by \textit{Proposition { 3 $\&$ 4}}, 
	\begin{equation}\label{ineqs}
	\tilde{u}^*(x_{\rm o})=\tilde{\mathcal{L}}\tilde{u}^*(x_{\rm o})\geq \mathcal{L}v^*(x_{\rm o},{x}_{\rm u})=v^*(x_{\rm o},{x}_{\rm u}) \mbox{ for any } x_u.
	\end{equation}
	Then,
	\begin{equation*}
    \begin{aligned}
        \bar{C}&\geq \tilde{u}^*(x_{\rm o})-\mathcal{L}\tilde{u}^*(x_{\rm o})\\
        &= \tilde{u}^*(x_{\rm o})-v^*(x_{\rm o},x_{\rm u})+\mathcal{L}v^*(x_{\rm o},x_{\rm u})-\mathcal{L}\tilde{u}^*(x_{\rm o}).
    \end{aligned}
	\end{equation*}
	Since $\mathcal{L}$ is a contraction mapping
    \begin{equation*}
        \begin{aligned}
        \mathcal{L}v^*(x_{\rm o},x_{\rm u})-\mathcal{L}\tilde{u}^*(x_{\rm o})\geq \beta \left(v^*(x_{\rm o},x_{\rm u})-\tilde{u}^*(x_{\rm o})\right).
        \end{aligned}
    \end{equation*}
    Therefore,
    \begin{equation*}
        \bar{C}\geq (1-\beta)\left(\tilde{u}^*(x_{\rm o})-v^*(x_{\rm o},x_{\rm u})\right)
    \end{equation*}
    Thus,
    \begin{equation*}
        \tilde{u}^*(x_{\rm o})-v^*(x_{\rm o},x_{\rm u})\leq \frac{\bar{C}}{1-\beta}.
    \end{equation*}
	By the same arguments, when $\tilde{u}^*(x_{\rm o})\leq \mathcal{L}\tilde{u}^*(x_{\rm o})$, and we acquire the similar inequality with bound $-\underline{C}/(1-\beta)$. Hence the proof of this theorem is completed.
\end{proof}
\qed

\begin{remark}
	The theorem above states that, even though the direct observation of ${x}_{\rm u}$ cannot be obtained, we can still guarantee that the performance is deteriorated at most by a bounded term.
\end{remark}

We note that the bound on the difference is dependent on the structure of the cost function with respect to $x_{\rm u}$. More explicitly, the bounds depend on the sensitivity of $c(x_{\rm o},x_{\rm u},a)$ with respect to the change in $x_{\rm u}$. 

Also, we can compare the value function in virtual belief method and the value function define in \eqref{opt3}. The comparison results are stated in the following theorem.

\begin{theorem}(\textbf{Comparison between POMDP and Virtual Belief Models})\\
	If the transition probability can be decomposed as
	\begin{equation*}
	p(x_{\rm o}^\prime,x_{\rm u}^\prime|x_{\rm o},x_{\rm u},a)=p(x_{\rm o}^\prime|x_{\rm o},a)p(x_{\rm u}^\prime|x_{\rm u},a),
	\end{equation*}
	then 
    \begin{equation*}
	\sup_{x_{\rm o, 0}}\left |	\inf_{\pi\in\Pi}V^\pi_{\alpha_{\rm u}}(x_{\rm o,0})  - \inf_{\pi\in\Pi} \tilde{V}^{\pi}_{\alpha_{\rm u}}(x_{{\rm o},0}) \right| 
    \end{equation*}
    is bounded by 
    \begin{equation*}
        C^\prime=\max\left\{\frac{-\underline{C}^\prime}{1-\beta},\frac{\bar{C}^\prime}{1-\beta}\right\},
    \end{equation*}
    where
	\begin{equation*}
	\bar{C}^\prime:=\max_{b,b^\prime\in\mathcal{B},x_{\rm o},a} \left \{ \bar{c}(x_{\rm o},b,a)-\bar{c}(x_{\rm o},b^\prime,a) \right \},
	\end{equation*}
	and
	\begin{equation*}
	\underline{C}^\prime:=\min_{b,b^\prime\in\mathcal{B},x_{\rm o},a} 
	\left \{ \bar{c}(x_{\rm o},b,a)-\bar{c}(x_{\rm o},b^\prime,a) \right \}.
	\end{equation*}
\end{theorem}

\begin{proof}
	The proof of \textit{Theorem 6} largely relies on the proof of \textit{Theorem 5}.
\end{proof}
\qed

The results in \textit{Theorem 5 $\&$ 6} hold under the assumption that the transition probability can be decomposed. Generally, when this assumption does not hold, even if the cost function does not change significantly with respect to $x_{\rm u}$, the results may not hold. In such cases, the update of the belief is required for estimating the evolution of observable part $\{x_{{\rm o},t}\}_t$.

%%%%%%%%%%%%%%%%%%%%%%%%%%%%%%%%%%%%%%%%%%%%%%%%%%%%%%%%%%%%%%%%%%%%%%%%%%%%%%%%

\section{SPECIAL CASES}

In this section, we discuss several special cases concerning the structure of the system dynamics, cost function, and transition probabilities.

\subsection{$x_{\rm u}=\bar{g}(x_{\rm o})$ or $\mathcal{X}_{\rm u}=\emptyset$}
If $x_{\rm u}=\bar{g}(x_{\rm o})$, the unobservable substate can be fully determined from the observable substate. That is, we can infer the true value of the unobservable state from the observation. Then, the overall state can be fully characterized by $x_{\rm o}$. Therefore, $x_{\rm o}$ is sufficient to represent the overall state of the system. As we mentioned earlier, in this case, the system reduces to MDP. It is straightforward to see that \eqref{virtualbellman} and \eqref{fullinfo} coincide and thus they yield the same value function. Similar arguments hold for the case where $\mathcal{X}_{\rm u}=\emptyset$.

\subsection{$\mathcal{X}_{\rm o}=\emptyset$}
In this case, the system is a deterministic system in which the state can be fully characterized by the belief $b$. And the approximated optimization faced here is given by 
\begin{equation}\label{classicalOC}
\begin{aligned}
    \min_{\pi}&\ \sum_{t=0}^{\infty}\beta^t b_t^{\rm T}c(a)\\
    \text{s.t.}&\ b_{t+1}=P_{\rm u}(a)b_{t}
\end{aligned}
\end{equation}
Here, with a slight abuse of notation, $c(a)=\left\{c(:,a)\right\}$ and $P_{\rm u}(a)$ is the transition matrix of $x_{\rm u}$ for a given action $a$. The optimization above is a classical nonlinear optimal control problem.

\subsection{${c}(x_{\rm o},{x}_{\rm u},a)={c}(x_{\rm o},{x}^\prime_{\rm u},a)$}
This case is trivial as the stage cost function is no longer a function of the unobservable substate. Thus, $\bar{C}=\underline{C}$. Here, $\bar{C}$ and $\underline{C}$ are defined in the proof of \textbf{Theorem 2}. The two value functions are equal and the virtual belief method loses no performance.

\subsection{$p({x}_{\rm o}^\prime,{x}_{\rm u}^\prime|{x}_{\rm o},{x}_{\rm u},a)=p({x}_{\rm o}^\prime|{x}_{\rm o},a)p({x}_{\rm u}^\prime|{x}_{\rm u})$}
In this case, the independent random process, $\left\{{x}_{{\rm u},t}\right\}_{t}$, is not controllable and thus evolves independently respect to the actions. By assuming that the transition kernel $p(x^\prime_{\rm u}|x_{\rm u})$ is ergodic, we denote the stationary measure of $x_{\rm u}$ by $b_{\rm s}(x_{\rm u})$ and the corresponding belief vector by $b_{\rm s}$. In such that a setting, as there exist stationary measures over $x_{\rm o}$ and $x_{\rm u}$ jointly, we can reduce \eqref{DP1} to
\begin{equation}\label{stat}
u_{\rm s}(x_{\rm o})=\min_{a}\bigg\{\bar{c}(x_{\rm o},b_{\rm s},a)+\beta\sum_{x_{\rm o}^\prime}u_{\rm s}(x_{\rm o}^\prime)p(x_{\rm o}^\prime|x_{\rm o},a)\bigg\},
\end{equation}
which leads to tractable linear programmings as the state space and number of constraints are finite. As for the virtual belief method, if we replace the initial belief vector $b_0$ with $b_{\rm s}$, then it will yield the optimal solution.

%%%%%%%%%%%%%%%%%%%%%%%%%%%%%%%%%%%%%%%%%%%%%%%%%%%%%%%%%%%%%%%%%%%%%%%%%%%%%%%%

\section{NUMERICAL EXAMPLE}
In this section, we use numerical experiments to demonstrate o results.
Consider the following dynamical system:
\begin{equation*}
	\mathcal{A}=\{0,1\},\ \mathcal{X}_{\rm o}=\{0,1\},\ \mathcal{X}_{\rm u}=\{0,1\},
\end{equation*}
 Let $P_{\rm o}(a)$ and  $P_{\rm u}(a)$ be the transition matrices of the observable substate and unobservable substate, respectively.
\begin{equation*}
\begin{aligned}
P_{\rm o}(0)=\begin{bmatrix}
0.8 & 0.2\\ 
0.5 & 0.5
\end{bmatrix},\quad 	P_{\rm o}(1)=\begin{bmatrix}
0.2 & 0.8\\ 
0.8 & 0.2
\end{bmatrix},
\end{aligned}
\end{equation*}
\begin{equation*}
\begin{aligned}
P_{\rm u}(0)=\begin{bmatrix}
0.2 & 0.8\\ 
0.8 & 0.2
\end{bmatrix},\quad 	P_{\rm u}(1)=\begin{bmatrix}
0.8 & 0.2\\ 
0.5 & 0.5
\end{bmatrix}.
\end{aligned}
\end{equation*}
Let $\mathcal{C}(a)=\{c(x_{\rm o},x_{\rm u},a)\}_{\mathcal{X}_{\rm o}\times\mathcal{X}_{\rm u}}$ be the cost matrix.
\begin{equation*}
	\mathcal{C}(0)=\begin{bmatrix}
	2 & 0.2\\ 
	1 & 0.4
	\end{bmatrix},\quad \text{and}\quad 
		\mathcal{C}(1)=\begin{bmatrix}
	0.7 & 0.4\\ 
	1 & 0.4
	\end{bmatrix}.
\end{equation*}
And the probabilities of $x_{{\rm u},0}=0$ and $x_{{\rm o},0}=0$ are both assumed to be 1/2. The discount factor is set as $\beta =1/2$.

First consider the system of full information, i.e., the agent has access to both $x_{\rm o}$ and $x_{\rm u}$. The optimal value found by solving LP is approximately $2.0524$. And if $x_{\rm u}$ cannot be observed, using virtual belief method and we have that the policy $d(x_{\rm o})=1$, $\forall\ x_{\rm o}$. and we obtain the value as $2.3714$, which is bounded by $C$.

Another way to find the find the stationary policy is by solving the following constrained linear programming:\\

\textbf{(Constrained) Dual LP$^*$}
\begin{equation*}
\begin{aligned}
\min_{y_{\rm c}}\quad &\sum_{a,x_{\rm o},x_{\rm u}}y_{\rm c}(x_{\rm o},x_{\rm u},a)c(x_{\rm o},x_{\rm u},a)\\
{\rm s.t.}\quad 
& \sum_{a}y_{\rm c}(x^\prime_{\rm o},x^\prime_{\rm u},a)-\alpha_{\rm o}(x^\prime_{\rm o})\alpha_{\rm u}(x_{\rm u}^\prime)\\
&\  =\sum_{x_{\rm o},x_{\rm u},a}\beta p(x^\prime_{\rm o},x^\prime_{\rm u}|x_{\rm o},x_{\rm u},a)y_{\rm c}(x_{\rm o},x_{\rm u},a),\quad \forall\ x^\prime_{\rm o},x_{\rm u}^\prime\\
& y_{\rm c}(x_{\rm o},x_{\rm u},a)=y_{\rm c}(x_{\rm o},x^\prime_{\rm u},a),\quad\forall\ x_{\rm o}, x_{\rm u},x^\prime _{\rm u},a.\\
\end{aligned}
\end{equation*}

Here, $y(x_{\rm o},x_{\rm u},a)$ is the measure function measuring the frequency of the system visiting the state-action pair $(x_{\rm o},x_{\rm u},a)$. The constraint arises from the fact that $x_{\rm u}$ cannot be observed and used in the policy in the system whose state is solely $x_{\rm o}$. The corresponding primal LP is given by\\

\textbf{Primal LP$^*$}
\begin{equation*}
\begin{aligned}
\min_{u_{\rm c}}&\sum_{x_{\rm o},x_{\rm u}}u_{\rm c}\left(x_{\rm o},x_{\rm u}\right)\alpha_{\rm o}(x_{\rm o})\alpha_{\rm u}(x_{\rm u})\\
\text{s.t.} &\sum_{x_{\rm u}}u_{\rm c}(x_{\rm o},x_{\rm u})-\sum_{x_{\rm u}}c(x_{\rm o},x_{\rm u},a)\\
&\quad\ \ \geq \beta \sum_{x_{\rm u},x_{\rm o}^\prime,x_{\rm u}^\prime}u_{\rm c}(x_{\rm o}^\prime,x_{\rm u}^\prime)
p(x_{\rm o}^\prime,x_{\rm u}^\prime|x_{\rm o},x_{\rm u},a), \quad \forall\ x_{\rm o},a.
\end{aligned}
\end{equation*}

Note that there exists one-to-one correspondence between the solutions to \textbf{Primal LP$^*$} and \textbf{Dual LP$^*$}. There is no dynamic programming equation associated with \textbf{Primal LP$^*$}, yet it provides a numerical method to compute the stationary policy. The optimal deterministic stationary found using the constrained LPs is given by $d_{\rm MD}(x_{\rm o})=1$, $\forall\ x_{{\rm o}}$, which yields value $1.8706$.

%%%%%%%%%%%%%%%%%%%%%%%%%%%%%%%%%%%%%%%%%%%%%%%%%%%%%%%%%%%%%%%%%%%%%%%%%%%%%%%%
\section{CONCLUSIONS AND FUTURE WORKS}

\subsection{Conclusions}

In this paper, we have studied LSI-MDP, which is a Markov decision process with incomplete state information. In this model, the state can be divided into two parts, one of which is observable and the other is unobservable. System of this kind is closely related to MDP and POMDP and we have pointed out their relations. We have shown that directly solving optimal control problem in such systems is challenging using the classical linear programming approach, as the number of decision variables (or constraints) is possibly infinite. We have proposed a new method to tackle this challenge, which provides a suboptimal solution. We have provided  bounds on  the difference between optimal and sub-optimal solution, under certain separability conditions.

\subsection{Future Works}
When coping with an optimization problem with uncertainty, the agent can either average the randomness out or be robust to the uncertainty. As a consequence, if we consider the unobservable substate as the source of uncertainty, we can formulate a robust optimal control problem, where the objection is given by the following:
\begin{equation*}
\begin{aligned}
V^{\pi}_{\rm R}(x_{{\rm o},0}) 
=\sup_{\mathbf{x}_u} \sum_{t=0}^\infty
\mathbb{E}^{\pi}\left[  \beta^tc(x_{{\rm o},t},x_{{\rm u},t},a_t)\bigg| \{ \mathbf{x}_{{\rm u},s} \}_{s \leq t} \right]. 
\end{aligned}
\end{equation*}
The agent aims to find an optimal policy while being robust to all the possible trajectories of the unobservable substate, $\mathbf{x}_{\rm u}$. It is worth noting that the  transition probabilities and the cost function uncertainty share the same uncertainty induced by $x_{\rm u}$, which makes the robust problem NP-hard as shown in \cite{mannor2012lightning,bagnell2001solving}. %

\appendix
\section{Proof of Theorem 1}    % Each appendix must have a short title.
\begin{proof}
    \begin{itemize}
        \item[a)] 	The assumptions imply that, given $x_{\rm o}$ and $x_{\rm u}$, we can uniquely determine $\tilde{x}_{\rm u}$ as $g(x_{\rm o},\cdot)$ is injective with respect to $\tilde{x}_{\rm u}$. Thus,
	\begin{equation*}
	\begin{aligned}
	p(x_{\rm o}^\prime,{x}_{\rm u}^\prime|x_{\rm o},x_{\rm u},a)&=p(x_{\rm o}^\prime,g(x_{\rm o}^\prime,\tilde{x}_{\rm u}^\prime)|x_{\rm o},g(x_{\rm o},\tilde{x}_{\rm u}),a)\\
	&=p(x_{\rm o}^\prime,\tilde{x}_{\rm u}^\prime|x_{\rm o},\tilde{x}_{\rm u},a).
	\end{aligned}
	\end{equation*}
	we obtain the following decomposition of the transition probability, 
    \begin{equation*}
    p(x_{\rm o}^\prime,x_{\rm u}^\prime|x_{\rm o},x_{\rm u},a)=p(x_{\rm o}^\prime|x_{\rm o},a)p(\tilde{x}_{\rm u}^\prime|\tilde{x}_{\rm u},a),
    \end{equation*}
    as $\tilde{x}_{\rm u}$ and $x_{\rm o}$ are independent. Therefore, we can transform the joint random process $(x_{\rm o},x_{\rm u})$ into two conditionally  independent random processes (conditioned on the sequence of decisions).
	\item[b)] We prove the the second part of the theorem by induction. At the initial time, $x_{{\rm o},0}$ is given and there exists a distribution of $x_{{\rm u},0}$, which is denoted by $\alpha_{\rm u}(x_{\rm u})$. We can express the unobservable state as a function $g(\cdot,\cdot)$ such that
	\begin{equation}\label{compose2}
	X_{{\rm u},0} = g(x_{{\rm o},0},\Tilde{X}_{{\rm u},0}),
	\end{equation}
	where $\Tilde{X}_{{\rm u},0}$ is a random variable that is conditionally independent of $X_{{\rm o},0}$. At time $t$, we assume that there exists $\Tilde{X}_{{\rm u},t}$ which is independent of $X_{{\rm o},t}$ such that 
	\begin{equation*}
	X_{{\rm u},t} = g(X_{{\rm o},t},\Tilde{X}_{{\rm u},t}).
	\end{equation*}
	Then, given the realizations of $X_{{\rm o},t}$ and $\Tilde{X}_{{\rm u},t}$, we can determine the value of $x_{{\rm o},t}$ and ${x}_{{\rm u},t}$. \\
	At time $t+1$, there exists a function $\tilde{f}$ such that
	\begin{equation*}
	\begin{aligned}
	(X_{{\rm o},t+1},X_{{\rm u},t+1})&=f(X_{{\rm o},t},X_{{\rm u},t},a,\Gamma)\\
	&=f(X_{{\rm o},t},g(X_{{\rm o},t},\Tilde{X}_{{\rm u},t}),a,\Gamma)\\
	&=\Tilde{f}(X_{{\rm o},t},\Tilde{X}_{{\rm u},t},a,\Gamma).
	\end{aligned}
	\end{equation*}
	Hence the second part of the theorem follows.
    \end{itemize}
\end{proof}
\qed

\bibliographystyle{splncs04}
\bibliography{biblio}
\end{document}